\documentclass[lettersize,journal]{IEEEtran}
\usepackage{amsmath,amssymb,amsfonts}
\usepackage{algorithmic}
\usepackage{algorithm}
\usepackage{array}
\usepackage[caption=false,font=footnotesize,labelfont=rm,textfont=rm]{subfig}
\usepackage{textcomp}
\usepackage{stfloats}
\usepackage{url}
\usepackage{verbatim}
\usepackage{graphicx}
\usepackage{cite}
\usepackage{bm}
\usepackage{booktabs}
\usepackage{soul, color, xcolor}
\usepackage{diagbox}
\sethlcolor{yellow}
\usepackage{multirow}
\usepackage{booktabs}
\usepackage{siunitx}

\newtheorem{theorem}{\textbf{Theorem}}
\newtheorem{proof}{\textbf{Proof}}

\usepackage{tikz,xcolor,hyperref}
\hypersetup{hidelinks,
	colorlinks=true,
	allcolors=black,
	pdfstartview=Fit,
	breaklinks=true}
\definecolor{lime}{HTML}{A6CE39}
\DeclareRobustCommand{\orcidicon}{%
    \begin{tikzpicture}
    \draw[lime, fill=lime] (0,0) 
    circle [radius=0.16] 
    node[white] {{\fontfamily{qag}\selectfont \tiny ID}};    \draw[white, fill=white] (-0.0625,0.095) 
    circle [radius=0.007];    \end{tikzpicture}
    \hspace{-2mm}}
\foreach \x in {A, ..., Z}{%
    \expandafter\xdef\csname orcid\x\endcsname{\noexpand\href{https://orcid.org/\csname orcidauthor\x\endcsname}{\noexpand\orcidicon}}
    }

\soulregister\cite7
\soulregister\ref7
\soulregister\eqref7
\begin{document}

\title{Parallelizable Complex Neural Dynamics Models for PMSM Temperature Estimation with Hardware Acceleration}

\author{Xinyuan Liao\orcidA{}, 
\IEEEmembership{Graduate Student Member, IEEE}, 
Shaowei Chen\orcidB{}, \IEEEmembership{Member, IEEE}, 
and Shuai Zhao\orcidC{}, \IEEEmembership{Senior Member, IEEE} 
\thanks{Xinyuan Liao is with the School of Electronics and Information, Northwestern Polytechnical University, Xi’an 710072, China, and also with the Department of Electrical and Electronic Engineering, The Hong Kong Polytechnic University, Hong Kong (e-mail: liaoxinyuan@mail.nwpu.edu.cn; xin-yuan.liao@connect.poylu.hk).}

\thanks{Shaowei Chen is with the School of Electronics and Information, Northwestern Polytechnical University, Xi’an 710072, China (e-mail: cgong@nwpu.edu.cn).}

\thanks{Shuai Zhao is with the AAU Energy, Aalborg University, Aalborg 9220, Denmark (e-mail: szh@energy.aau.dk).}}

\markboth{Submitted to IEEE}%
{Shell \MakeLowercase{\textit{Liao et al.}}: Parallelizable Complex Neural Dynamics Models for PMSM Temperature Estimation with Hardware Acceleration}

\IEEEpubid{0000--0000/00\$00.00~\copyright~2021 IEEE}

\maketitle
\begin{abstract}
Accurate and efficient thermal dynamics models of permanent magnet synchronous motors are vital to efficient thermal management strategies. Physics-informed methods combine model-based and data-driven methods, offering greater flexibility than model-based methods and superior explainability compared to data-driven methods. Nonetheless, there are still challenges in balancing real-time performance, estimation accuracy, and explainability. This paper presents a hardware-efficient complex neural dynamics model achieved through the linear decoupling, diagonalization, and reparameterization of the state-space model, introducing a novel paradigm for the physics-informed method that offers high explainability and accuracy in electric motor temperature estimation tasks. We validate this physics-informed method on an NVIDIA A800 GPU using the JAX machine learning framework, parallel prefix sum algorithm, and Compute Unified Device Architecture (CUDA) platform. We demonstrate its superior estimation accuracy and parallelizable hardware acceleration capabilities through experimental evaluation on a real electric motor.
\end{abstract}

\begin{IEEEkeywords}
System Thermal Dynamics, State-Space Models, Control-Oriented Modeling, Physics-Informed Machine Learning, Parallel Computing.
\end{IEEEkeywords}

\section{Introduction}
Permanent magnet synchronous motors (PMSMs) are widely used in various industrial applications, including electric vehicles (EVs) and renewable energy power generation \cite{car, jestie0}, due to their outstanding power output and reliability. However, their maximum torque, power density, and overall health are highly sensitive to temperature variations \cite{review}. An effective thermal management system is essential for enhancing both power density and operational lifespan \cite{failure, jestie1, jestie2}. Despite its importance, thermal management in PMSMs presents significant challenges. The complex structure of the stator and rotor, along with high-speed rotation, makes it difficult to instrument physical temperature sensors (e.g., PT100) in critical areas without impacting system performance. As a result, indirect temperature measurement through a motor thermal dynamics model is often necessary for real-time monitoring and control. Moreover, thermal control strategies such as model predictive control (MPC), which serve as the foundation of many thermal management systems \cite{MPC}, rely heavily on accurate predictive models, referred to here as motor thermal dynamics models. Achieving a balance between real-time performance and estimation accuracy remains a long-standing challenge in the industry \cite{benchmark}.

The motor thermal dynamics modeling methods primarily encompass model-based and data-driven methods. Model-based methods typically leverage thermodynamic or mathematical principles to construct a thermodynamic model of PMSM for analyzing its thermal characteristics \cite{review}. Excellent real-time performance is the main advantage of model-based methods, while it frequently raises concerns regarding their accuracy and efficiency. Their complex modeling procedure is one of the major challenges in practical implementation. Data-driven methods mainly focus on artificial intelligence (AI) technology \cite{zhao2020overview, GNN, motor,pcim}, without extensive expert knowledge. These methods generally rely on black-box architectures, which impedes industrial implementation. In industrial scenarios, the trustworthiness of a prediction model is more important than the prediction accuracy alone. Furthermore, large model size is required for data-driven models to effectively capture complex data patterns, resulting in poor real-time performance. Given the limited computing capability of computing units, it is challenging to implement the data-driven method in the field.

\IEEEpubidadjcol

Physics-informed methods as a compromise between model-based and data-driven methods have inherited both advantages \cite{zhao2022parameter}. Physical priors can regularize the data-driven models' output. For instance, Liao et al. \cite{liao} utilized physics-informed neural networks (PINN) to constrain the output of the surrogate model to comply with the underlying dynamic principles in lifetime evolution, having augmented the explainability and performance of data-driven models. Moreover, physical priors can guide the design of the foundational structure of data-driven models. Wallscheid et al. \cite{LPTN} proposed the lumped-parameter thermal network (LPTN) by merging the heat conduction equation with empirical data-based model parameter identification methods. Kirchg{\"a}ssner et al. \cite{TNN} further enhanced estimation accuracy and model versatility by integrating the heat conduction equation with the state-space model (SSM). Liao et al. \cite{liao2025neural} embed the perron-frobenius theorem into the SSM, to ensure the system stability of the data-driven method. J{\'a}n et al. \cite{building} utilize the overall heat transfer coefficient of the building to constrain the eigenvalues of SSM. 

Although physics-informed methods inherit the model-based method's high explainability and the data-driven method's flexibility, they also inherit the drawbacks of the above two methods. Due to physics properties being different between different equipment, industrial implementation needs a certain extent of experience to establish a physics-informed model properly, which often brings limitations for applying the physics-informed methods. Moreover, different types of prior knowledge often yield challenges for knowledge integration within a model \cite{phm}. In addition, the real-time performance of physics-informed methods is still drastically behind the white-box model-based methods, which is crucial in practical implementation. With the advancement of computing devices, such as graphics processing units (GPUs), developing a hardware-efficient method holds significant potential for large-scale industrial applications. 

Inspired by neural dynamics models \cite{liao2025neural} and linear RNNs \cite{deepmind}, this paper proposed a parallelizable complex neural dynamics model (complexNDM) based on SSM, for system modeling with temperature estimation as a case study. In the complex domain, we have archived embedding various general physical priors (stability and system oscillation frequency) into the data-driven model structure and established a bridge between these physical priors and the property of smooth evolution on the non-chaos system. Moreover, due to the diagonalizability of matrices in the complex domain, the method is parallelizable and we achieved hardware parallel acceleration of complexNDM on the compute unified device architecture (CUDA) platform. Finally, the proposed method is applied to PMSM temperature estimating, which shows that the method’s performance is superior to the existing methods in estimating the temperature of the motor system. The main contributions are summarized as follows.
\begin{enumerate}
    \item A parallelizable complex neural dynamics model is proposed for the PMSM temperature estimation, achieving a high estimation accuracy with a compact model size. The code details accompanying the paper are open-sourced on GitHub\footnote{[Online]. https://github.com/XinyuanLiao/ComplexNDM}.
    \item A more transparent physics-informed machine learning framework that integrates various physical prior information is proposed, demonstrating the method for stability, system oscillation frequency, and smooth evolution characteristics embedding.
    \item The proposed framework is parallelizable and can be effectively accelerated by GPU. The time complexity is reduced from $O(N)$ to $O(log_2N)$ by parallel implementation.
\end{enumerate}

The remainder of this paper is organized as follows. Section \ref{back} introduces the basic framework and method modules involved in this paper. Section \ref{parameter} introduces the methods of diagonalization and physical information embedding. Section \ref{anal} provides a detailed analysis of parameterized diagonal state spaces. Section \ref{exp} introduces the experimental verification and analyzes the model estimation performance, hardware acceleration, and model eigenvalues. Section \ref{cons} summarizes the conclusions of the paper.

\section{Methodology}\label{back}
\subsection{State-Space Model}

State-space models (SSMs) have been extensively applied in modeling physical systems due to their straightforward structure and properties. It represents a physical system as a set of inputs, outputs, and internal states. The relationship between inputs, outputs, and internal states can be described by first-order differential equations. For a discrete time-invariant system, its SSM is expressed as
\begin{equation}
\begin{aligned}
x_{t}&=A x_{t-1}+B u_{t}, \\
y_t&=C x_t+D u_t,
\end{aligned}
\end{equation}
where $A$ is the state matrix, $B$ is the input matrix, $C$ is the output matrix, and $D$ is the feedforward matrix. $D$ denotes the residual connection linking the input to the output. In general, physical systems exhibit a time lag between inputs and outputs, so $D$ is often assumed to be $0$.

The complexity of industrial systems often exceeds the capabilities of basic linear SSMs, which rely solely on linear transformations, leading to inaccuracies in modeling. Integrating nonlinear modules into the state-space model significantly enhances its expressive capacity, enabling accurate representation of these intricate systems.

\subsection{Structured Linear Neural Dynamics Model}
Modeling dynamic systems by neural networks can be referred to as the neural dynamic model \cite{building}, with its representation form as
\begin{equation}
\begin{aligned}
& x_0=f_0\left(\left[y_{1-n} , \ldots , y_0\right]\right), \\
& x_t=f_x\left(x_{t-1}, u_t\right), \\
& y_t=f_y\left(x_t\right),
\end{aligned}
\end{equation}
where $f_0$, $f_x$, and $f_y$ are nonlinear mappings. This black-box representation draws inspiration from recurrent neural networks (RNNs), where inputs and states exhibit coupling among each other. However, recent studies \cite{deepSSM,deepmind} have indicated that RNNs can be more precise by additive decoupling the input from state transition utilizing the presentation of SSM. Essentially, the SSM represents a linear variant of RNNs. These linear RNN-based models \cite{deepmind, deepSSM} not only overcome the limitations of vanilla RNN serial computations but also exhibit significant advantages when handling longer series. Therefore, Liao et al. \cite{liao2025neural} decoupled and linearized the black-box neural dynamics model into a  gray-box form as
\begin{equation}
\begin{aligned}
x_0 &=f_0\left(\left[y_{1-n} , \ldots , y_0\right]\right), \\
x_t &=A x_{t-1}+f_u\left(u_t\right), \\
y_t &=C x_t.
\label{eq: dynamics}
\end{aligned}
\end{equation}
where $A$ is the state matrix, $C$ is the output matrix, and $f_u$ nonlinear maps the control input into the state space. The linear update method of the neural dynamics model naturally has the potential for parallel computing, as
\begin{equation}
\begin{aligned}
x_1 &=A x_0+f_u(u_1), \\
x_2 &=A\left(A x_0+f_u(u_1)\right)+f_u(u_2), \\
&=A^2 x_0+\left(A f_u(u_1)+f_u(u_2)\right), \\
x_k &=A^k x_0+ \sum_{j=1}^kA^{k-j} f_u(u_{j}).
\label{eq: parallel}
\end{aligned}
\end{equation}

\begin{figure}[t]
    \centering
    \includegraphics[width=\linewidth]{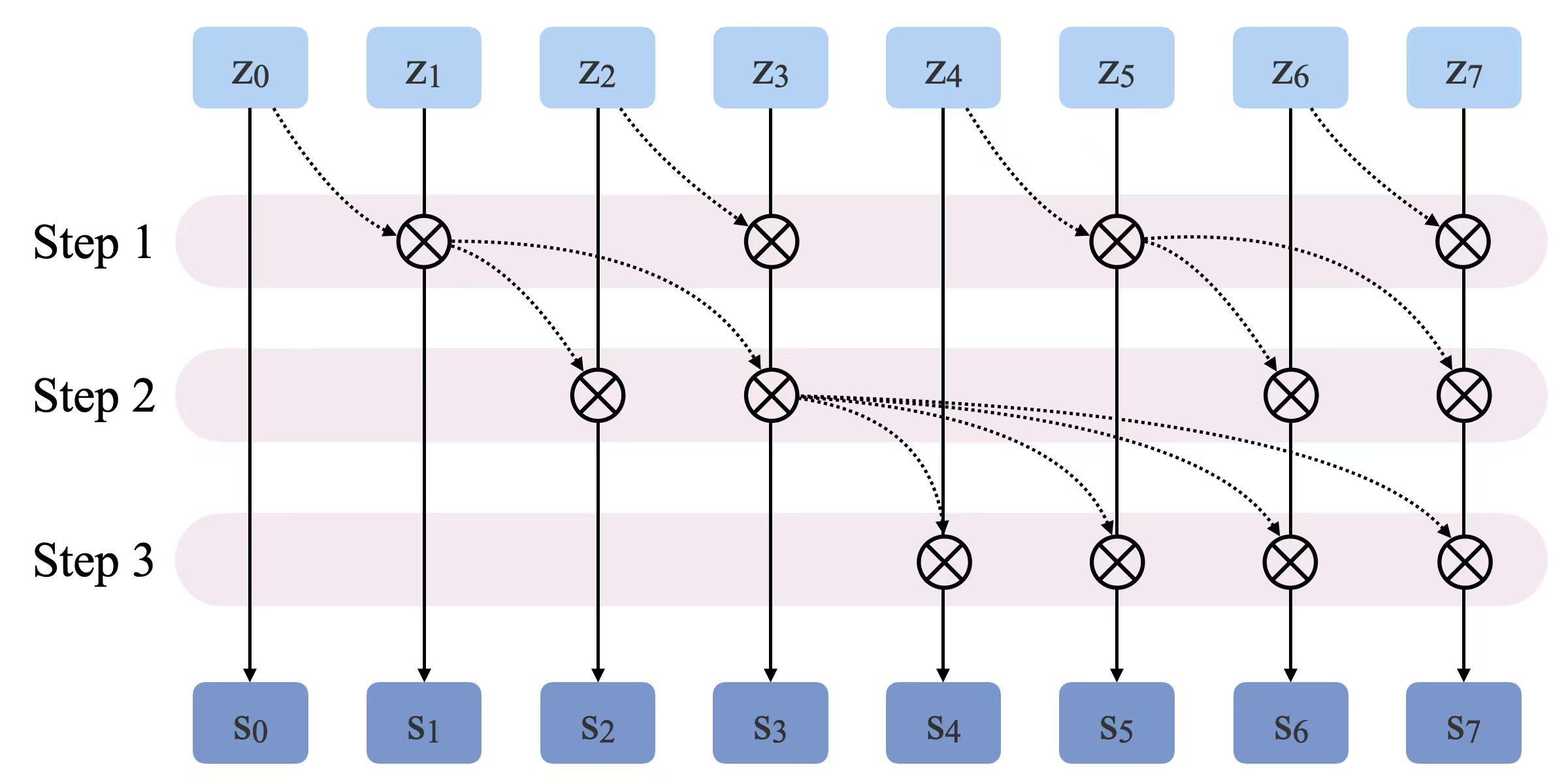}
    \caption{Pipeline of the parallel prefix sum algorithm, where $z_i$ is the input and $s_i$ denotes the $\sum_{j=0}^iz_i$ when the binary operator is plus.}
    \label{fig: prefix}
\end{figure}

As in (\ref{eq: parallel}), the hidden state $x_k$ does not depend on the state $x_{k-1}$ at the previous moment directly. With an initial state $x_0$ and a sequence of control inputs, the SSM-based neural dynamics model can compute the state at any moment in parallel. It is worth noting that although the state transfer process is assumed to be a linear process, the calculation of the hidden state $x_k$ still involves a large number of nonlinear operators $f_u$, which are generally constructed as a nonlinear neural network. Therefore, the assumption of linear state transition in this case does not violate the principles of real-world nonlinear systems.
\begin{figure*}[t]
  \centering
    \includegraphics[width=\linewidth]{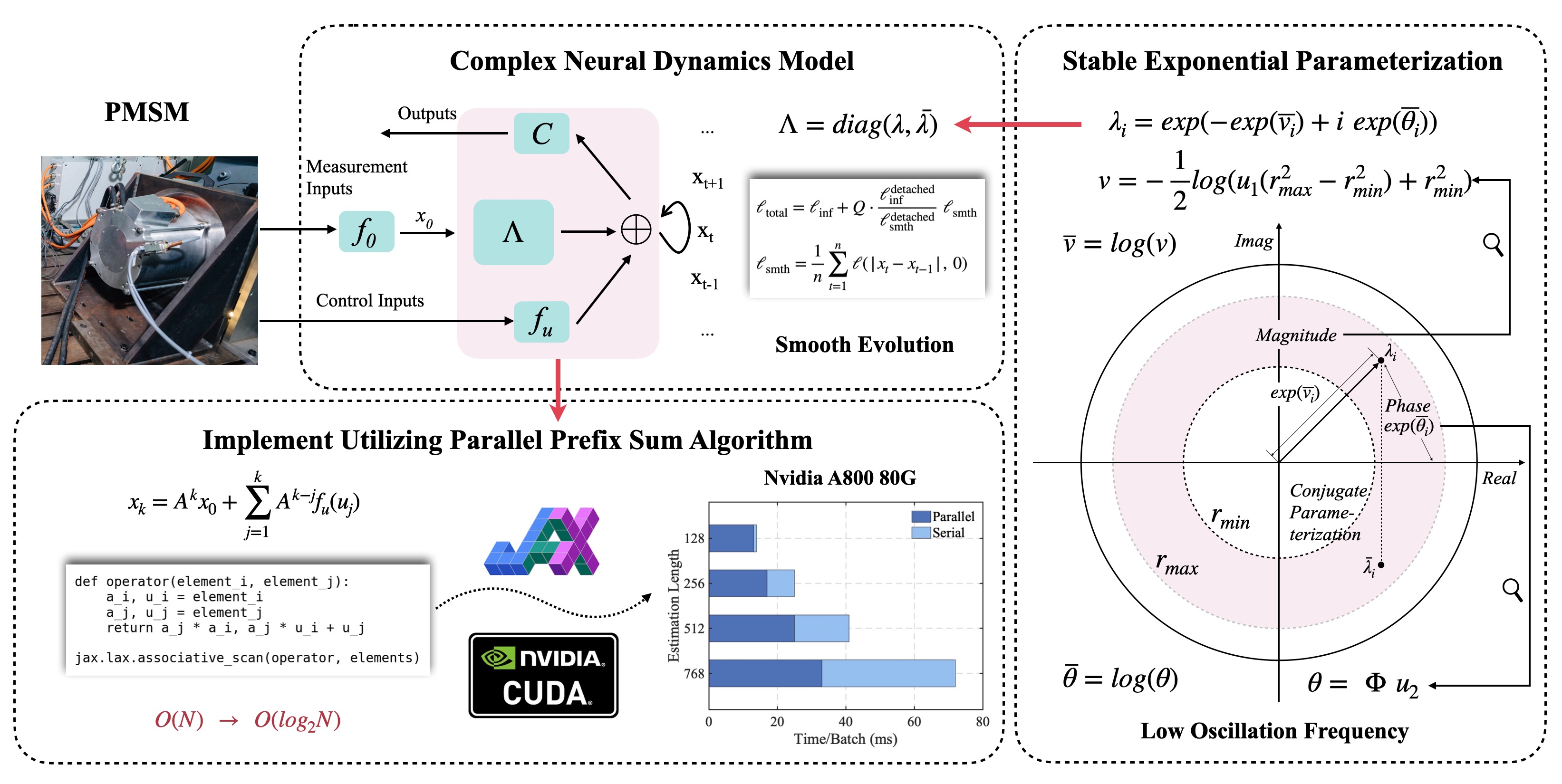}
    \caption{Structure of the Complex Neural Dynamics Model (complexNDM). Given the PMSM's \cite{motor} previous temperature measurements to estimate the initial state and parallel compute system outputs by control inputs. The loss function contains a priori bias towards a smooth evolution of non-chaotic systems. Parallel computing is based on the parallel prefix sum algorithm (Scan), which can be implemented in GPUs by CUDA. The complex diagonal parameterization of the state matrix incorporates the physical priori information concerning system stability and low oscillation frequency.}
    \label{fig: frame}
\end{figure*}
\subsection{Parallel Prefix Sum Algorithm}

In deep learning, parallelization methods are divided into data parallelism and model parallelism. Data parallelism is the most common parallelization strategy supported by almost all deep learning frameworks, which accelerates the training and inference process by splitting the data into multiple small batches and processing them simultaneously on multiple computing nodes. Model parallelism splits the model into different parts and calculates them in parallel on different computing resources. Data parallelism can accelerate any model based on ANNs, while model parallelism requires the model to have a specific structure, such as the Transformer. Vanilla RNN cannot use model parallelism for training and inference acceleration because of its nonlinear recursive structure, but the assumption of linear state transfer allows RNN to be accelerated with the help of model parallelism. 

The parallel prefix sum algorithm is a typical method in model parallelism. Given any binary operator $\otimes$ and a sequence $(z_0,z_1,\cdot\cdot\cdot,z_{n-1})$, it can efficiently compute the sequence $(s_0,s_1,\dots, s_{n-1})=(z_0,(z_0\otimes z_1),\dots,(z_0\otimes z_1\otimes\dots\otimes z_{n-1}))$ in parallel. It leverages multi-processor hardware, such as GPUs, to attain parallel acceleration by employing the divide-and-conquer principle. Fig. \ref{fig: prefix} illustrates the execution flow of the parallel prefix sum algorithm. For a sequence of length $N$, serial calculation necessitates looping through $N$ steps, whereas parallel calculation requires only a logarithmic number $2\lceil log_2N \rceil$ of sequential steps \cite{prefix}. As a result, the time complexity of the method is significantly reduced from $O(N)$ to $O(log_2N)$.

The equation (\ref{eq: parallel}) can be parallelized using the parallel prefix-sum algorithm. The term $A^k x_0$ in (\ref{eq: parallel}) corresponds to a prefix-product problem, where the binary operator $\otimes$ represents matrix multiplication. In this context, $s_i$ in Fig.~\ref{fig: prefix} corresponds to $A^i$ with the fixed coefficient $x_0$. The summation term $\sum_{j=1}^k A^{k-j} f_u(u_j)$ can be regarded as a weighted prefix-sum problem over the index $j \in [1,k]$. This is because it can be recursively decomposed into two parallel sub-problems: $A^{k-j} \sum_{i=1}^jA^{j-i} f_u$ and $\sum_{i=j+1}^kA^{k-i} f_u\left(u_i\right)$, and further split in the same manner. In this case, the input $z_i$ in Fig.~\ref{fig: prefix} corresponds to $A^{k-j} f_u(u_j)$ in (\ref{eq: parallel}), while $s_i$ corresponds to $\sum_{j=1}^k A^{k-j} f_u(u_j)$. It is worth noting that the parallel algorithm achieves its highest efficiency when the splitting index satisfies $j = (i+k)/2$.

\subsection{Complex Neural Dynamics Model}

Based on the structured linear neural dynamics model, this paper proposed a novel physics-informed machine learning framework for PMSM temperature estimation termed as complex neural dynamics models, and Fig. \ref{fig: frame} presents the details of complexNDM. The complexNDM archived embedding the stability and low oscillation into the data-driven structure by a special parameterization and initialization in the complex domain, as detailed in Section \ref{parameter}. In addition, the constraint of the smooth evolution property is achieved by introducing an extra loss term into the loss function. Finally, we have significantly accelerated the complexNDM by diagonalizing the state matrix and the parallel prefix sum algorithm.

\section{Parameterizing Diagonal State Spaces}\label{parameter}

\subsection{Diagonalization}
\begin{figure}
    \centering
    \includegraphics[width=0.9\linewidth]{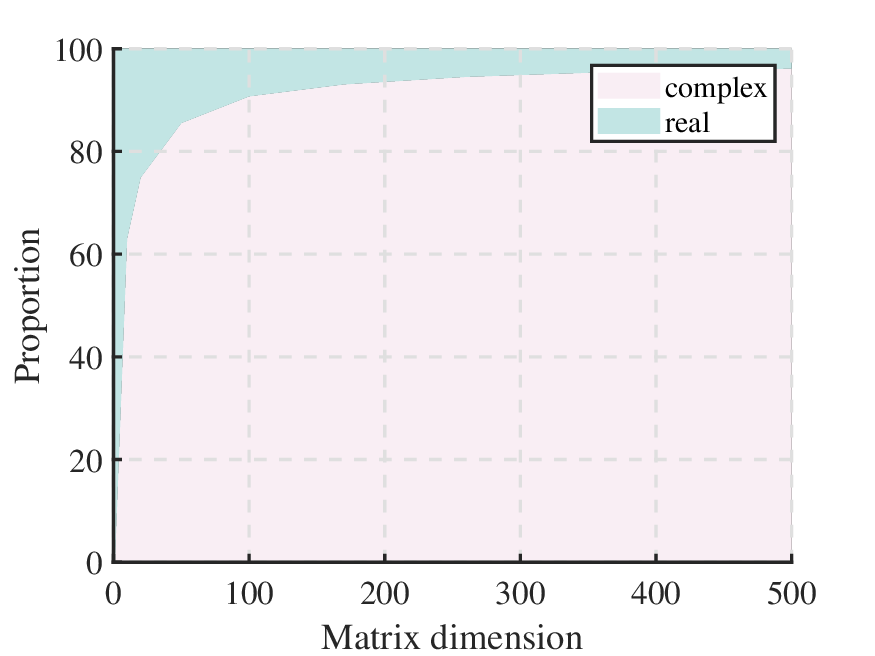}

    \caption{The average proportion of real and complex values in the eigenvalues of a random square matrix.}

    \label{fig: diag}
\end{figure}

The main calculation effort of (\ref{eq: parallel}) comes from the exponentiation of the state matrix $A$. It is naturally conceivable that if $A$ can be transformed into a diagonal matrix, the calculation effort of (\ref{eq: parallel})  will be significantly reduced by altering the matrix multiple to the vector multiple. It is well known that the probability of a real square matrix being diagonalizable in the real number domain decreases as the dimension of the matrix increases. Fig. \ref{fig: diag} illustrates the trend between the proportion of complex numbers in the eigenvalues of matrices with increasing dimension. In other words, as the matrix dimension goes up, the proportion of complex eigenvalues increases to 1. Therefore, it is difficult to complete the diagonalization of a matrix only in the real number domain. However, up to an arbitrarily small perturbation of the entries, every matrix $M \in \mathbb{R}^{n\times n}$ is diagonalizable \cite{axler2024linear}, i.e., one can write $M=P \Lambda P^{-1}$, where $P \in \mathbb{C}^{n \times n}$ is an invertible matrix and $\Lambda = \text{diag}(\bm{\lambda}) \in \mathbb{C}^{n \times n}$. In other words, the set of non-diagonalizable matrices has measure zero, see e.g. Zhinan (2002) \cite{zhinan2002jordan} for a proof idea.


Assuming that the dimension of the state matrix $A$ is an even number and that all of its eigenvalues are complex numbers. It is important to note that the state represented by $A$ is a hidden state, which does not directly correspond to the actual physical state. The output matrix $C$, on the other hand, establishes a mapping to the actual physical quantities. Therefore, the assumption that the number of states is even does not introduce any significant issues regarding the model's universality. According to repetitive experiments, when the matrix dimension surpasses 20, the proportion of complex numbers in its eigenvalues already exceeds 74\% according to Fig. \ref{fig: diag}, for example. Hence, this assumption above will not decline the model's performance significantly, the final diagonal form is given as

\begin{equation}
\begin{aligned}
    A=P \Lambda P^{-1},\ \Lambda =
\begin{bmatrix}
    \text{diag}(\bm{\lambda}) & \bm{0} \\
    \bm{0} & \text{diag}(\bm{\bar{\lambda}}) \\
\end{bmatrix}.
    \label{eq: finaldiag}
\end{aligned}
\end{equation}
where $\bm{\lambda}$ is the set of unique eigenvalues of the state matrix $A$ and $\bm{\bar{\lambda}}$ is the conjugate of $\bm{\lambda}$. Plugging (\ref{eq: finaldiag}) into (\ref{eq: parallel}), (\ref{eq: parallel}) changed to
\begin{equation}
\begin{aligned}
P^{-1} x_k =\Lambda^k P^{-1} x_0+\sum_{j=1}^k \Lambda^{j-1} P^{-1} f_u\left(u_{k-j}\right).
\end{aligned}
\end{equation}
Letting $\hat{x_i}=P^{-1} x_i$ and $\hat{f_u}=P^{-1} f_u$, we can get 
\begin{equation}
    \hat{x_k} =\Lambda^k \hat{x_0}+\sum_{j=1}^k  \Lambda^{j-1}  \hat{f_u}\left(u_{k-j}\right).
\end{equation} 
After diagonalization, (\ref{eq: dynamics}) changed to
\begin{equation}
\begin{aligned}
\hat{x_0} &=\hat{f_0}\left(\left[y_{1-n} , \ldots , y_0\right]\right), \\
\hat{x_t} &=\Lambda \hat{x_{t-1}}+\hat{f_u}\left(u_t\right), \\
y_t &=\mathcal{R}(\hat{C} \hat{x_t}).
\end{aligned}
\end{equation}
where $\mathcal{R}(\hat{C} \hat{x_t})$ denotes the real part of $\hat{C} \hat{x_t}$ \cite{gupta2022diagonal}, $\hat{C}=CP$, and $\hat{f_0}=P^{-1} f_0$. Note that since $\hat{C} \in \mathbb{C}^{m\times n}$ and the output layer of networks are complex-valued, the projection matrix $P^{-1}$ can be merged into $\hat{C}$ and networks without any extra computations. Assuming the function of the output layer within the $f_0$ and $f_u$ is $y=wx^\text{T}+b$, $P^{-1}(wx^\text{T}+b)=P^{-1}wx^\text{T}+P^{-1}b$. Evidently, when $w\in \mathbb{C}^{n\times m}, b \in \mathbb{C}^n$, the projection matrix $P^{-1}$ can be omitted in the training process.

\subsection{Stability Constraint for Data-driven Pipeline}
The stability of the data-driven pipeline is the key to ensuring model safety in field implementation. The internal state of an unstable model undergoes continual expansion, resulting in the accumulation of errors and consequently leading to elevated estimation inaccuracies. Unstable parameterization will lead to difficulty in model training \cite{deepmind}. Moreover, employing an unstable model as a predictive model within a control system may yield unreliable predictions, resulting in catastrophic failure.

Since $\hat{x_k} =\Lambda^k \hat{x_0}+\sum_{j=1}^k  \Lambda^{j-1}  \hat{f_u}\left(u_{k-j}\right)$, the exponentiation of the diagonal matrix $A$ is the source of instability in model training and inference. The hidden state $x_k$ will explode or vanish exponentially as $k$ increases, and spectral analysis \cite{stable} can better demonstrate this phenomenon.

Assume $\bm{\lambda}$ is the set of eigenvalues of matrix $M$, then the spectral radius of $M$ is $\rho(M) = max(|\bm{\lambda}|)$. A sufficient condition to guarantee the stability of the neural dynamics model is $\rho(\Lambda) < 1$, where $\Lambda$ denotes the state matrix \cite{deepSSM}. For a diagonal matrix, the elements on its diagonal are its eigenvalues. Compared with the stability constraint method for the dense matrix \cite{stable}, the stability constraint method for the diagonal matrix is more direct and intuitive.

Let $u_1, u_2$ be uniform independent random variables in the interval $[0,1]$, $\Phi$ be the upper bound of the phases, and $0<r_{min}<r_{max}<1$. Compute the magnitude $v=-\frac{1}{2}log(u_1(r_{max}^2-r_{min}^2)+r_{min}^2)$ and phase $\theta=\Phi u_2$. As indicated in \cite{deepmind}, $\bm{\lambda}=e^{-v+i\theta}$ will be uniformly distributed on the ring in $\mathbb{C}$ between circles of radii $r_{min}$ and $r_{max}$.

\begin{theorem}   
    Let $\overline{v}=log(v)$, the exponential of a complex number $e^{-v+i\theta}=exp(-exp(\overline{v})+i\theta)$. And $|exp(-exp(\overline{v})+i\theta)|=1$ is achieved at $\overline{v}=-\infty$, while $|exp(-exp(\overline{v})+i\theta)|=0$ is achieved at $\overline{v}=\infty$.
    \label{th: 0-1}
\end{theorem}
\begin{proof}
    Since $e^{-v+i\theta}=e^{-v}(cos\theta+i~sin\theta)$, $|e^{-v+i\theta}|=\sqrt{e^{-2v}cos^2\theta+e^{-2v}sin^2\theta}=e^{-v}$, then $e^{-v}=1$ is achieved at $\overline{v}=-\infty$ and $e^{-v}=0$ is achieved at $\overline{v}=\infty$.
\end{proof}

The theorem above gives a way to constrain the eigenvalues of $\Lambda$ to $[0,1]$, and close to 1. In detail, $\overline{v}$ and $\theta$ are the learnable parameters to calculate the unique eigenvalues $\bm{\lambda}$ of the matrix $\Lambda$ (It's well-known that the remaining eigenvalues are the conjugate complex numbers of $\bm{\lambda}$). According to Theorem \ref{th: 0-1}, regardless of the values of $\overline{v}$ and $\theta$, $\bm{\lambda}$ always falls within the range of 0 to 1. Moreover, to make the state matrix $\bm{\lambda}$ align more closely with physical principles, it is imperative to initialize the values of $r_{min}$ and $r_{max}$ reasonably. From a thermodynamics perspective \cite{building}, for example, $\rho(\Lambda)$ is loosely related to the overall heat transfer coefficient of the system. In detail, \cite{building} sets the upper and lower bounds of $\rho(\Lambda)$ to $[0.8,1.0]$ for stability and low dissipativity of learned dynamics. While this correlation requires expert knowledge, it is more efficient and convenient to determine the initialization values from the standpoint of RNNs \cite{deepmind}, aiming for a spectral radius close to 1. Fig. \ref{fig: init} displays the distribution of $32$ eigenvalues randomly generated on the complex plane when $r_{min}$ is set to 0.4, $r_{max}$ to 0.8, and $\Phi$ to $\pi/4$.

\begin{figure}[t]
    \centering
    \includegraphics[width=\linewidth]{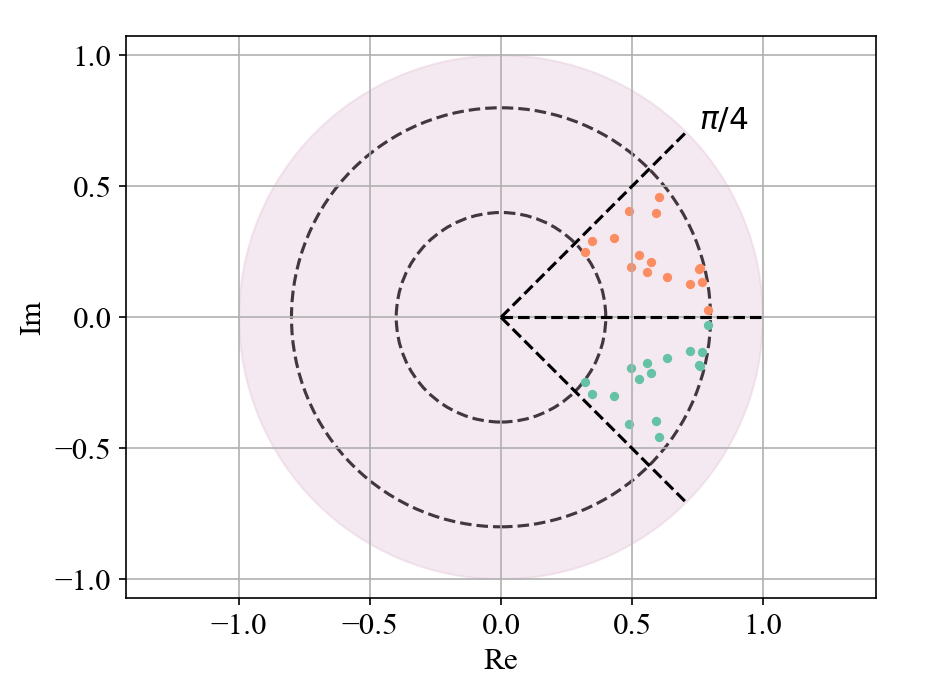}
    \caption{The distribution of randomly generated $\bm{\lambda} \in \mathbb{C}^{32}$ on the complex plane when $r_{min}=0.4$, $r_{max}=0.8$, and $\Phi=\pi/4$. Note that the eigenvalue with a negative imaginary part is the complex conjugate of an eigenvalue with a positive imaginary part.}
    \label{fig: init}
\end{figure}
\subsection{System Oscillation Frequency}
It is well-known that the imaginary component of the eigenvalue of the state matrix in the SSM corresponds to the system's periodic oscillation frequency. However, this physical prior is usually ignored. For example, in a system thermal model represented by a first-order RC-network (combination of thermal resistances and capacitors), the oscillatory behavior of the system cannot be displayed in principle. Introducing the oscillatory behavior of the system into the data-driven model not only makes the data-driven model more consistent with physical reality but also broadens the scope of application of this method. 

In most cases, when a system is primarily influenced by control input, its periodic oscillation frequency tends to be small \cite{building}. Consequently, the imaginary component of the eigenvalues associated with the state matrix should also be restricted. Considering the PMSM as an example, its temperature is predominantly influenced by external inputs rather than exhibiting periodic changes. Uniformly initializing the eigenvalue phase of the system thermal dynamics state matrix inherently biases the network towards learning spurious features in the input sequence. A state matrix with large phases will lead to the state oscillating violently in the complex plane during the training process, which will push the model to focus on capturing the mean result of local oscillation patterns \cite{deepmind}. This phenomenon also aligns with the frequency principle (F-Principle) \cite{xu2020frequency} observed in deep neural networks (DNNs), which indicates that DNNs often fit target functions from low to high frequencies during the training process. This bias revealed the reason why initializing models with a high oscillation frequency, i.e., a large phase, can result in suboptimal performance. Based on the analysis above, it is essential to embed this prior bias when initializing the model. Specifically, it is achieved by restricting the range of $\theta$ to a thin slice around 0. As the optimal values of $\theta$ are close to 0, we optimize the logarithmic phase $\overline{\theta}=log(\theta)$ instead of optimizing $\theta$ directly.

\subsection{Smooth Evolution Regulation}
It is well-known that the evolution of a non-chaotic system is smooth, implying that a huge difference between hidden states in consecutive moments is unrealistic. With this physical prior knowledge, the difference in hidden states between consecutive time steps can be utilized to regularize the model so that the system state transition is smooth. The new loss function can be expressed as
\begin{equation}
    \begin{aligned}
        \ell_{\rm total}&=\ell_{\rm inf}+Q\cdot \frac{\ell_{\rm inf}^{\rm detached}}{\ell_{\rm smth}^{\rm detached}}\cdot \ell_{\rm smth},\\
        &=\ell(\hat{y},y)+Q\cdot \frac{\ell_{\rm inf}^{\rm detached}}{\ell_{\rm smth}^{\rm detached}}\cdot\frac{1}{n}\sum _{t=1}^n \ell(|x_{t}-x_{t-1}|,\bm{0}).
    \end{aligned}
\end{equation}
where $\ell_{\rm inf}$ is the inference loss, $\ell_{\rm smth}$ is the smooth loss, $\hat{y}$ is the estimated value, $y$ is the real value, $x_t$ is the hidden state of the system at time $t$, $x_{t-1}$ is the hidden state of the system at time $t-1$, $\ell$ is an arbitrary loss function, and $Q$ is the weight coefficient of the smooth loss term after unified scale. Note that $\rm detached$ means stop computing gradients, enabling unifying the scale of two losses without disturbing the backpropagation.

\section{Analysis of the Initialization}\label{anal}

Geometrically, the magnitude of the state matrix eigenvalues signifies the extent of expansion or contraction within the state. When the magnitude is below $1$, it indicates a contraction in the state, signifying a gradual convergence toward a stable state. It is consistent with Lyapunov stability. Any initialized state will converge toward the attractor, and states farther away from the attractor will converge faster. This phenomenon aligns with the Lyapunov global asymptotic exponential stability \cite{lyapunov}, and its theoretical proof is intuitively evident due to the state matrix being diagonal. It is clear that $|S_i*\Lambda_{ii}*...*\Lambda_{ii}|=|S_i|*|\Lambda_{ii}|*...*|\Lambda_{ii}|=|S_i|*|\Lambda_{ii}|^t$, where $S_i$ is the $i$-th state in the hidden states, and $\Lambda_{ii}$ is the corresponding state transition coefficient in state matrix $\Lambda$. Without accounting for the system input, since arbitrary $|\Lambda_{ii}|<1$ in terms of any initial value of $S_i$, it will eventually converge to the origin after a sufficient number of state transitions. The complexNDM not only fulfills numerical stability but also satisfies global asymptotic exponential stability from the perspective of dynamic systems.

\begin{figure}[t]
  \centering
    \includegraphics[width=\linewidth]{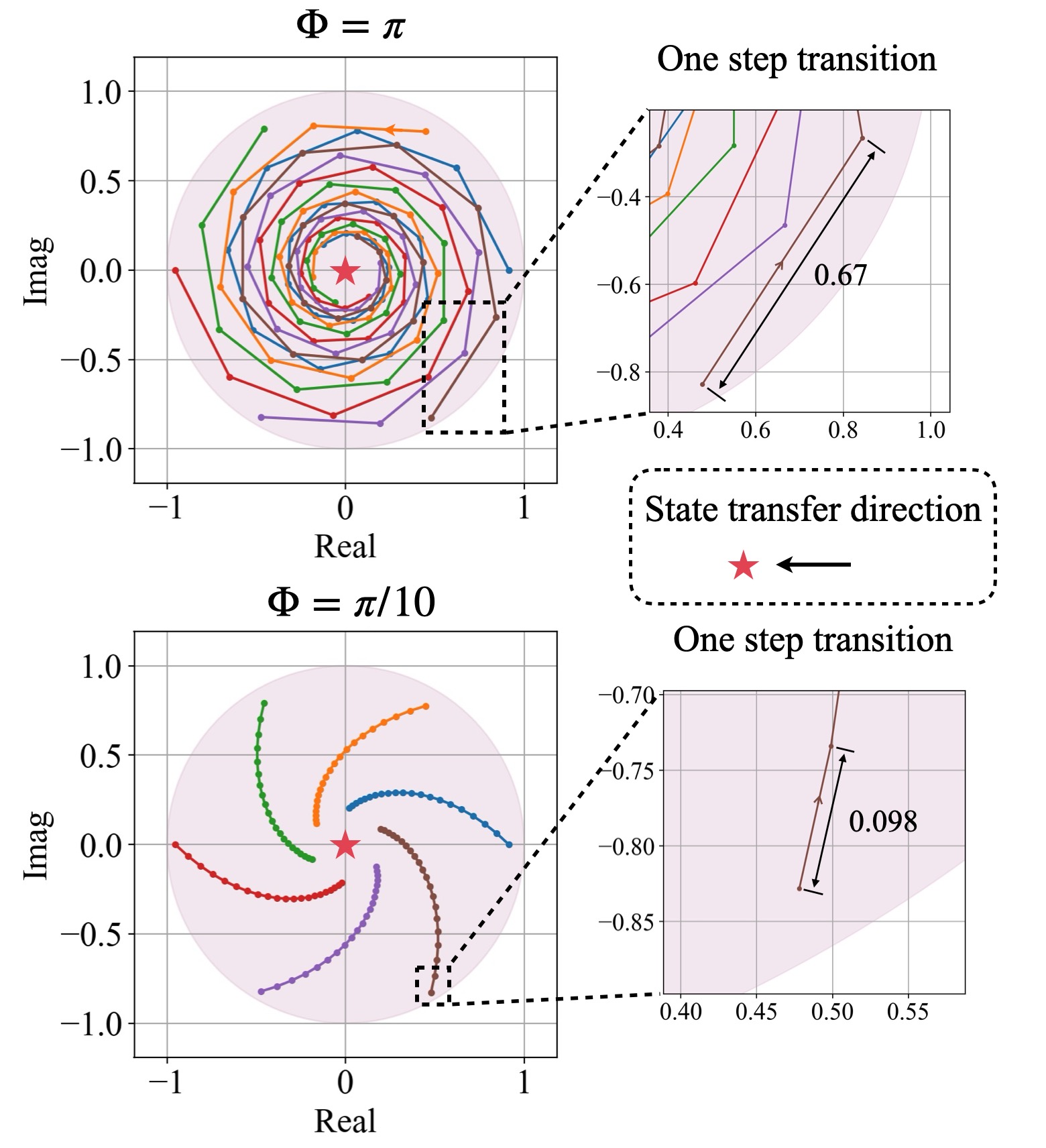}
    \caption{The upper figure is the state transition trajectories without system input when $\Phi=\pi$ and the lower figure is the same state transition trajectories without system input when $\Phi=\pi/ 10$.}
    \label{fig: trans}
\end{figure}

Moreover, the phase angle of the state matrix's eigenvalue represents the rotation angle of the state transition from a geometrical perspective. A smaller phase angle implies the state transition has a longer cyclical period and a slower rotational pace. Therefore, the above constraints on the magnitude approaching $1$ and the phase angle approaching $0$ are also the emphasis on the smooth evolution characteristics of the non-chaotic system. The complex multiplication formula under polar coordinate axes can clearly show this case as:
\begin{equation}
    (r_1e^{i\cdot\theta_1}) \cdot (r_2e^{i\cdot\theta_2}) = (r_1\cdot r_2)\cdot e^{i\cdot (\theta_1+\theta_2)},
\label{eq: transition}
\end{equation}
where $r_1$ is the magnitude of one item in the state matrix, $\theta_1$ is the corresponding phase, $r_2$ is the magnitude of the corresponding state, and $\theta_2$ is the phase of this state. Eq. (\ref{eq: transition}) denotes a state transition step, clearly demonstrating the above. The disparities in state transition trajectories under different phases are visually apparent in Fig. \ref{fig: trans}, offering an intuitive understanding of the phase constraint. Even when magnitudes are comparable, the phase will have a significant impact on the step size of each state transition. A smaller phase results in a smoother transition of states.

\section{Experiment and Implementation Analysis}\label{exp}
\subsection{Hardware setup and dataset}

The performance of the proposed method is demonstrated by using a PMSM temperature test bench, as shown in Fig. \ref{fig:test_bench}. The hardware setup for motor temperature data collection can be found in \cite{LPTN}. The dataset encompasses 185 h of multi-sensor data sampled at 2 Hz from a three-phase automotive traction PMSM rated at 52 kW, installed on a test bench. The motor operates under torque control, while its speed is regulated by a speed-controlled load motor rated at 210 kW and fed by a two-level IGBT inverter (Semikron: 3xSKiiP 1242GB120-4DW). The test torque ranges from -240 Nm to 260 Nm, and the motor speed ranges from 0 rpm to 6000 rpm. All measurements were captured by dSPACE analog-digital converters, synchronized with the control task. To ensure comparative analysis, this paper designates data from \textit{profile\_id}=58 as the validation set and data from \textit{profile\_id}s 65 and 72 as the test set, adhering to the specifications outlined in \cite{motor}. Table \ref{tab: dataset} shows the details of the dataset. All experiment results were completed on an Ubuntu server with JAX, an Intel(R) Xeon(R) Gold 6348 CPU @ 2.60GHz, and an NVIDIA A800 GPU.

\begin{figure}
    \centering
    \includegraphics[width=1\linewidth]{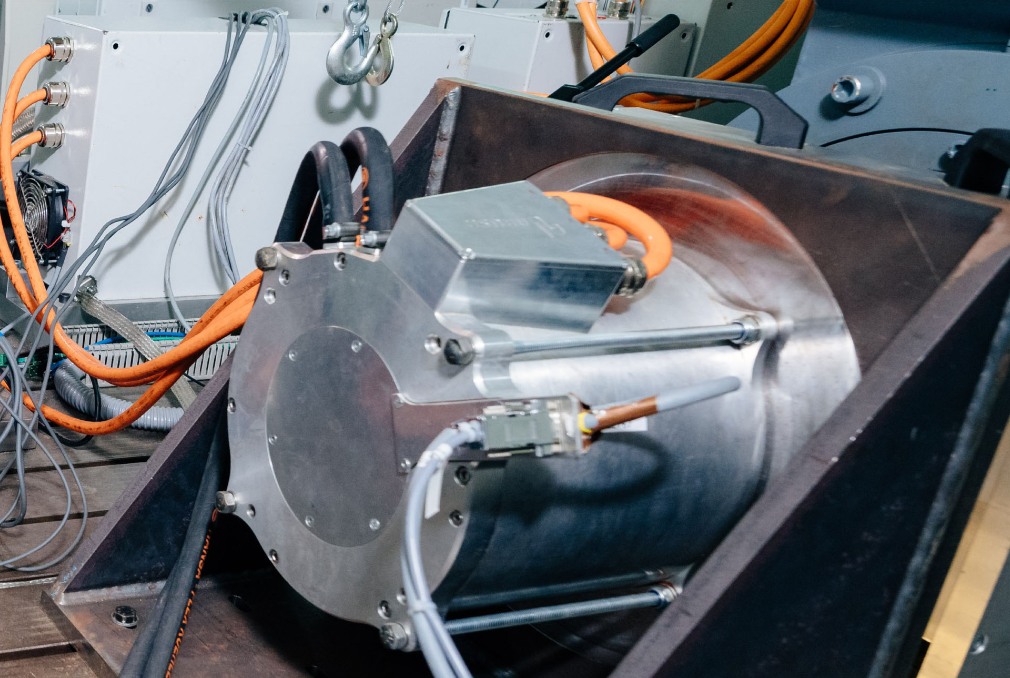}
    \caption{Test bench with an exemplary PMSM for data generation \cite{motor}.}
    \label{fig:test_bench}
\end{figure}

\begin{table}[t]
    \centering
    \caption{Dataset details}
    \begin{tabular}{l c|l c}
         \toprule
         Parameter&Symbol&Parameter&Symbol  \\
         \midrule
         \multicolumn{2}{c|}{\textbf{Measured Inputs}}&\multicolumn{2}{c}{\textbf{Derived Inputs}} \\
         Ambient Temperature& $\theta_a$&Voltage Magnitude&$u_s$ \\
         Liquid Coolant tempearture&$\theta_c$&Current Magnitude&$i_s$ \\
         Actual Voltage $d/q$-axes&$u_d,u_q$&\multicolumn{2}{c}{\textbf{Measured Target}} \\
         Actual Current $d/q$-axes&$i_d,i_q$&Permanent Magnet&$\theta_{PM}$ \\
         Torque&$T$&Stator Teeth&$\theta_{ST}$ \\
         Motor Speed&$S$&Stator Winding&$\theta_{SW}$ \\
         &&Stator Yoke&$\theta_{SY}$ \\
         \bottomrule
    \end{tabular}
    \label{tab: dataset}
\end{table}
\subsection{Data Preprocessing}
In the case of a non-chaotic system characterized by slow system state changes, a sampling frequency of 2 Hz holds little significance. Therefore, we downsample the data to a sampling frequency of 0.25 Hz. This action serves to alleviate the computational workload while facilitating the discernment of evolving data patterns. Specifically, data within each downsampled window is averaged to produce the new dataset. Moreover, the model training benefits from faster convergence when utilizing normalized data. This paper normalized temperature data and other data, respectively. All temperature data is only divided by 100 $^\circ C$ for normalization. As for other sensor data, division by the maximum absolute value is performed within the dataset for normalization.


\subsection{Hyperparameter Optimization and Ablation Study}
The training strategy adopts the linear warmup followed by a cosine decay scheme, where the initial learning rate (lr) is 1e-7, peak lr is 2e-4, end lr is 1e-7, and warm-up is in the first 10\% of training steps. The number of training epochs is selected as 300. For more robust training, Smooth$L_1$Loss \cite{l1smooth} is selected as the loss function during the training process. The prediction length is the length of measurement data input to the $f_0$, and the estimation length is the recurrent steps of the complexNDM, in other words, the estimation length is the length of the control input. The estimation length is set equally to \cite{liao2025neural} for comparison as 128. $f_0$ and $f_u$ are neural network modules that can be chosen freely. This paper chooses the basic multi-layer perceptron (MLP) with a complex-valued output layer as the backbone of the models, and the self-gated (Swish) \cite{swish} function is used to activate them. \textit{Hidden Layers} is the number of hidden layers of MLPs, and \textit{Hidden State Size} refers to the number of neurons in the hidden layer of MLPs. $f_0$ and $f_u$ both contain 2 hidden layers, and the hidden state size is designed as 32.

\begin{table}[h]
    \centering
    \caption{RMSE with different magnitude and phase}
    \begin{tabular}{ccccc}
        \toprule
        \multirow{2}{*}{Phase $\Phi$} & \multicolumn{4}{c}{Magnitude $[r_{min}, r_{max}]$} \\
        \cmidrule(lr){2-5}
        & $[0.0, 0.5]$&$[0.0, 1.0]$&$[0.5, 1.0 ]$&$[0.9, 1.0]$  \\
        \midrule
        $0.1\cdot \pi$ & $5.69\pm0.21$ & $1.77\pm0.21$ & $1.28\pm0.13$ & \bm{$1.03\pm0.04$}   \\
        $0.5\cdot \pi$ & $6.16\pm0.49$ & $3.89\pm1.30$ & $3.68\pm1.43$ & $2.37\pm0.95$   \\
        $1.0 \cdot \pi$ & $6.79\pm1.07$ & $4.82\pm1.63$ & $4.31\pm1.69$ &  $3.25\pm1.48$ \\
        \bottomrule
    \end{tabular}
    \label{tab: hpo}
\end{table}

The initialization of the range of magnitude and phase are the key hyperparameters in this work. We have thoroughly discussed the impact of different initialization on model performance, and the experimental results are shown in Table \ref{tab: hpo}. The experimental results reveal that the model exhibits the poorest performance when $r_{max}=0.5$, indicating that eigenvalue magnitudes below this threshold fail to satisfy the model's convergence criteria. In contrast, when $r_{max}=1.0$, we observe a progressive enhancement in model performance (both in terms of mean and standard deviation) as $r_{min}$ increases. This suggests that optimal convergence is achieved when all eigenvalues maintain magnitudes below yet approaching 1.0. Furthermore, the results demonstrate that the phase component of eigenvalues significantly influences model performance, with expanding phase ranges consistently leading to performance degradation (evident in both mean and standard deviation metrics). These empirical findings are in complete alignment with our theoretical analysis of eigenvalue magnitude and phase characteristics presented in the preceding analysis. As a result, we choose the range of magnitude $[r_{min},r_{max}]$ as $[0.9,1.0]$, and the upper bound of phase as $0.1\cdot\pi$.

\begin{figure}[h]
    \centering
    \includegraphics[width=\linewidth]{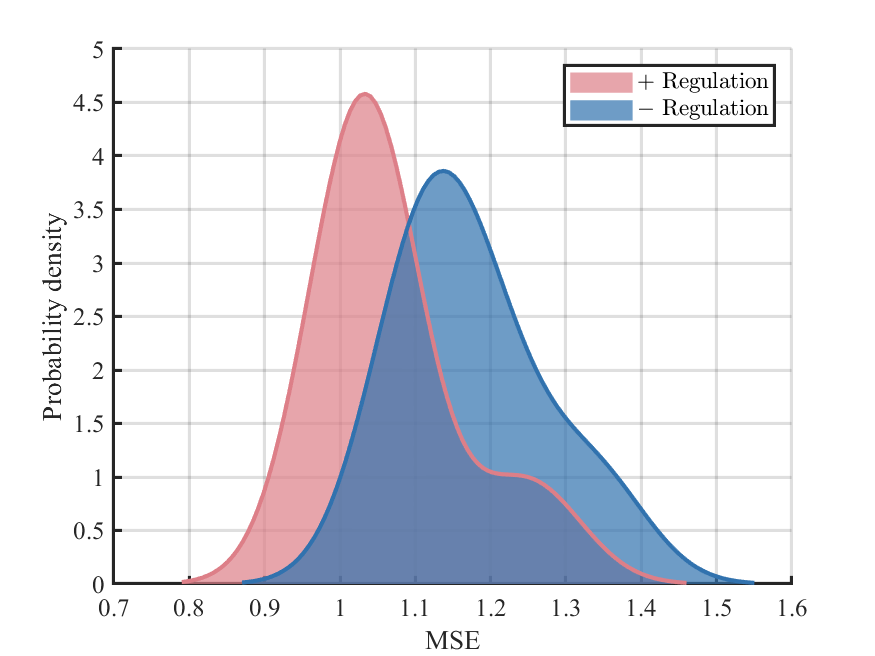}
    \caption{Ablation study for the Smooth Evolution Regulation.}
    
    \label{fig: ablation}
\end{figure}

The principle behind selecting the loss coefficient $Q$ is to maintain a smooth loss term smaller than the inference loss. This approach ensures that the smooth loss functions as a regularization factor will not disrupt the primary training trajectory of the model. This paper also conducted an ablation study about smooth evolution regulation by setting the loss coefficient $Q$ as 0.0 and 0.1. Fig. \ref{fig: ablation} presents the probability distribution curves of MSE obtained via kernel density estimation (KDE) for the results of 6 independent experiments using different random seeds. The estimation performance of complexNDM with smooth evolution regulation surpasses that without it in terms of both mean and variance. Specifically, the mean and variance for complexNDM with smooth evolution regulation are 1.07 and 0.0064, respectively, compared to 1.18 and 0.0081 without regulation. These results demonstrate that the complexNDM constrained by the smooth evolution regulation provides more stable, accurate, and reliable estimations under random scenarios.

The final model structure hyperparameters and training process hyperparameters were optimized via grid search methodology and are detailed in Table \ref{tab: hyper}.

\begin{table}[t]
    \centering
    \caption{Hyperparameters of the complexNDM}
    
    \begin{tabular}{lc|lc}
    \toprule
         Structural Parameters&Value&Training Parameters&Value  \\
         \midrule
         Estimation Length&$128$&Batch Size&$1024$ \\
         Prediction Length&$8$&Training Epochs& $300$\\
         Hidden State Size&$32$&Phase Upper Bound $\Phi$&$0.1\cdot\pi$\\
         Hidden Layers&$2$&$[r_{min}, r_{max}]$&$[0.9, 1.0]$\\ 
         Activation& Swish &Loss Coefficient ($Q$)&$0.1$ \\
         \bottomrule
    \end{tabular}
    \label{tab: hyper}
\end{table}

\begin{figure*}[t]
  \centering
  \includegraphics[width=\linewidth]{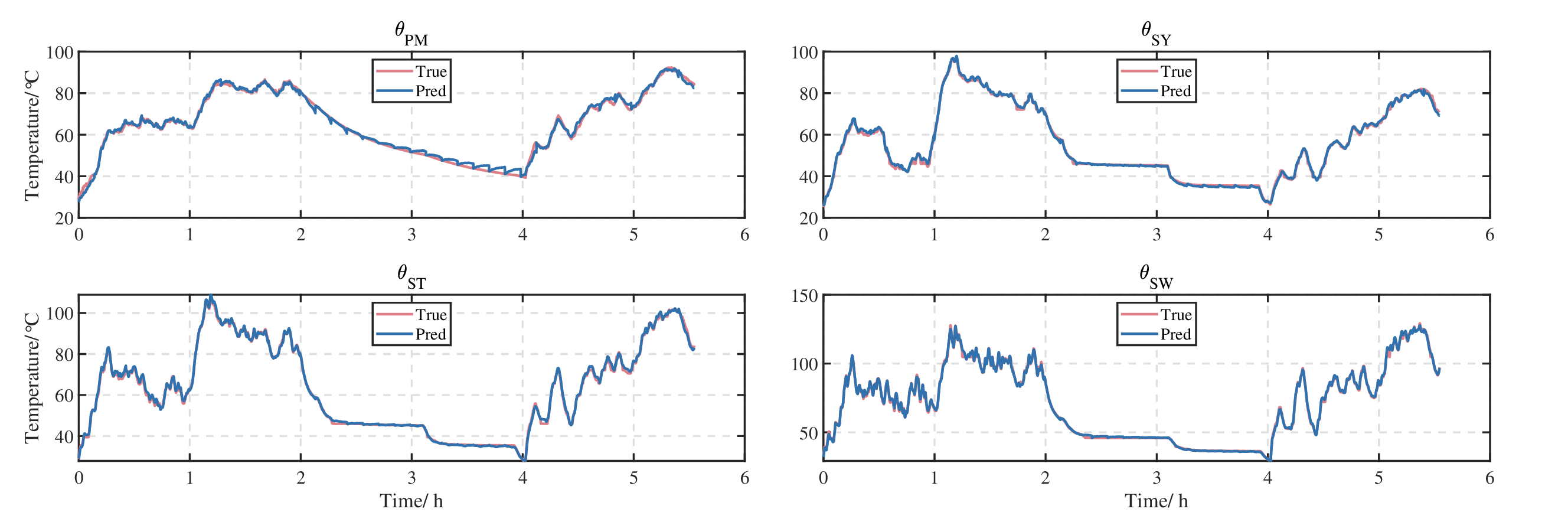}
  \caption{The comparison of the estimation and the ground truth of the temperature by complexNDM at $profile\_{id}=65$ in the test set.}
  \label{fig: pred}
\end{figure*}

\begin{table*}[h]
    \centering
    \caption{The performance and model size of existing methods}
    \begin{tabular}{p{5.5cm}|p{2.2cm}p{2.2cm}p{1.5cm}r}
    \toprule
         Method & RMSE ($K$) & MSE ($K^2$) & $\ell_{\infty}$ ($K$) & Parameters \\
    \midrule
         MLP (Kirchgässner et al. 2021) \cite{benchmark} & $2.36$ & $5.58$ & $14.3$ & $1{,}380$ \\
         OLS (Kirchgässner et al. 2021) \cite{benchmark} & $2.11$ & $4.47$ & $9.85$ & $328$ \\
         CNN (Kirchgässner et al. 2020) \cite{TNN} & $2.10$ & $4.43$ & $15.5$ & $4{,}916$ \\
         LTPN (Wallscheid et al. 2016) \cite{LPTN} & $1.91$ & $3.64$ & $7.37$ & $\bm{34}$ \\
         TNN (Small) (Kirchgässner et al. 2023) \cite{TNN} & $1.78$ & $3.18$ & $5.84$ & $64$ \\
         TNN (HPO) (Kirchgässner et al. 2023) \cite{TNN} & $1.69$ & $2.87$ & $6.02$ & $1{,}525$ \\
         RNN (Kirchgässner et al. 2021) \cite{motor} & $1.74$ & $3.02$ & $9.10$ & $>850k$ \\
         TCN  (Kirchgässner et al. 2021) \cite{motor} & $1.31$ & $1.72$ & $7.04$ & $>320k$ \\
         pfNDM (Liao et al. 2025) \cite{liao2025neural} & $\bm{0.95}$ & $\bm{0.91}$ & $6.59$ & $15.2k$ \\
    \midrule
         \textbf{complexNDM (This Work)} & $1.03\pm0.04$ & $1.07\pm0.08$ & $\bm{5.63\pm0.63}$ & $8{,}032$ \\
    \bottomrule
    \end{tabular}
    \label{tab: comp}
\end{table*}

\subsection{Performance Analysis} \label{sct: performance}

During testing, a sliding-window technique is employed to estimate temperatures across the monitoring duration. This window moves by the estimation length at each step. Fig. \ref{fig: pred} displays the prediction outcomes for the session with $profile\_id$=65 in the test set. The results visually demonstrate the remarkable prediction accuracy achieved by complexNDM. In addition, it can be seen that when the motor temperature drops evenly in the 2nd to 4th hour, the estimated value of the permanent magnet temperature has some slight jitter compared with the ground truth, while the estimated value of the stator temperature is consistent with the ground truth. This may be caused by the inconsistency of the thermodynamic characteristics of the stator and rotor.

\begin{figure}[t]
  \centering
    \includegraphics[width=\linewidth]{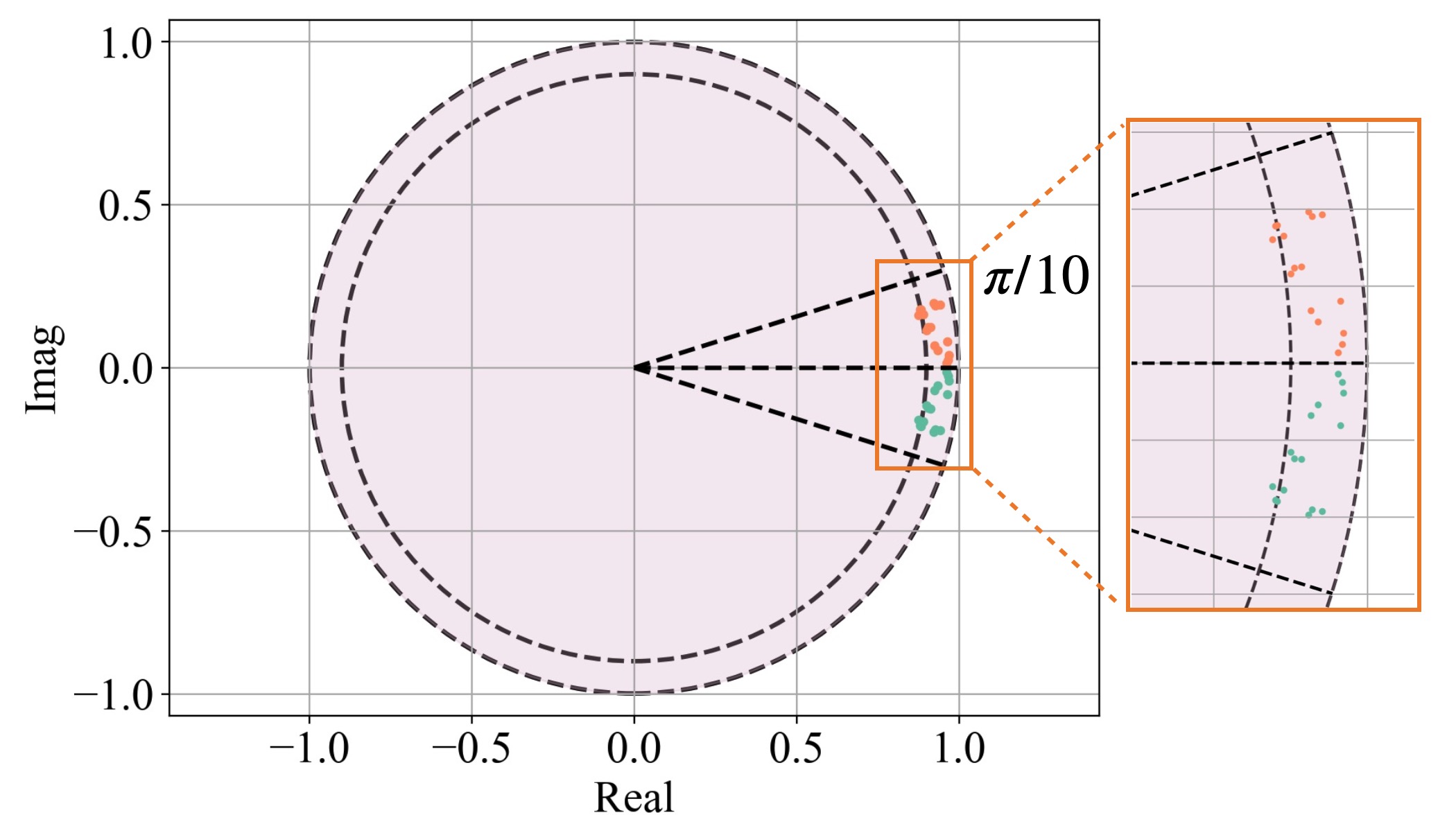}

    \caption{The distribution of eigenvalues within the trained complexNDM, which is initialized by $[r_{min},r_{max}]=[0.9,1.0]$ and $0.1\cdot\pi$.}

    \label{fig: eigen}
\end{figure}

This study quantifies the estimation performance of complexNDM and juxtaposes it with other existing methods, as detailed in Table \ref{tab: comp}. Statistical indicators such as MSE and RMSE show that its performance is only behind pfNDM\cite{liao2025neural}, but with nearly half the model size while having similar performance. And this method has fewer model parameters than other DNN-based methods. Moreover, due to the linear decoupling and diagonalization characteristics of complexNDM, most of the parameters are concentrated in the control input mapping network $f_u$, and the state matrix representing the main dynamic modes of the hidden state actually contains 32 parameters, which provides a basis for future research on system dynamics through model topology. Furthermore, complexNDM demonstrates superior performance in terms of the error's infinite norm. This advantage stems from the phase prior imposed on the model. A smoother state evolution enhances the model's capability to effectively manage extreme operating conditions. While the average error is important for estimation accuracy, the significance of the infinite norm surpasses it. A smaller maximum error value signifies the model's enhanced capability to manage extreme operating conditions effectively, thus promoting more efficient thermal management.

\subsection{Eigenvalue Analysis}

Fig. \ref{fig: eigen} illustrates the distribution of eigenvalues of the trained network on the complex plane. After training, the eigenvalue magnitudes are all less than 1, and the phases, which are close to 0, indicate that complexNDM is a stable model consistent with the characteristics of dynamics systems. Additionally, it is observable that the network's eigenvalues do not significantly deviate from the bias introduced during initialization. This implies that the prior information embedded during initialization persists throughout the network training process.

Furthermore, from the perspective of the dynamic mode, we can find that the model structure of complexNDM is efficient \cite{building}. The magnitude of all eigenvalues of complexNDM approximates 1, indicating each decoupled state is associated with an important dynamic mode \cite{dynamicmode}. This eigenvalue distribution shows that complexNDM achieved an efficient low-dimensional presentation of a dynamics system by diagonal decomposition.

\subsection{Hardware Acceleration}

\begin{figure}[t]
  \centering
  \includegraphics[width=\linewidth]{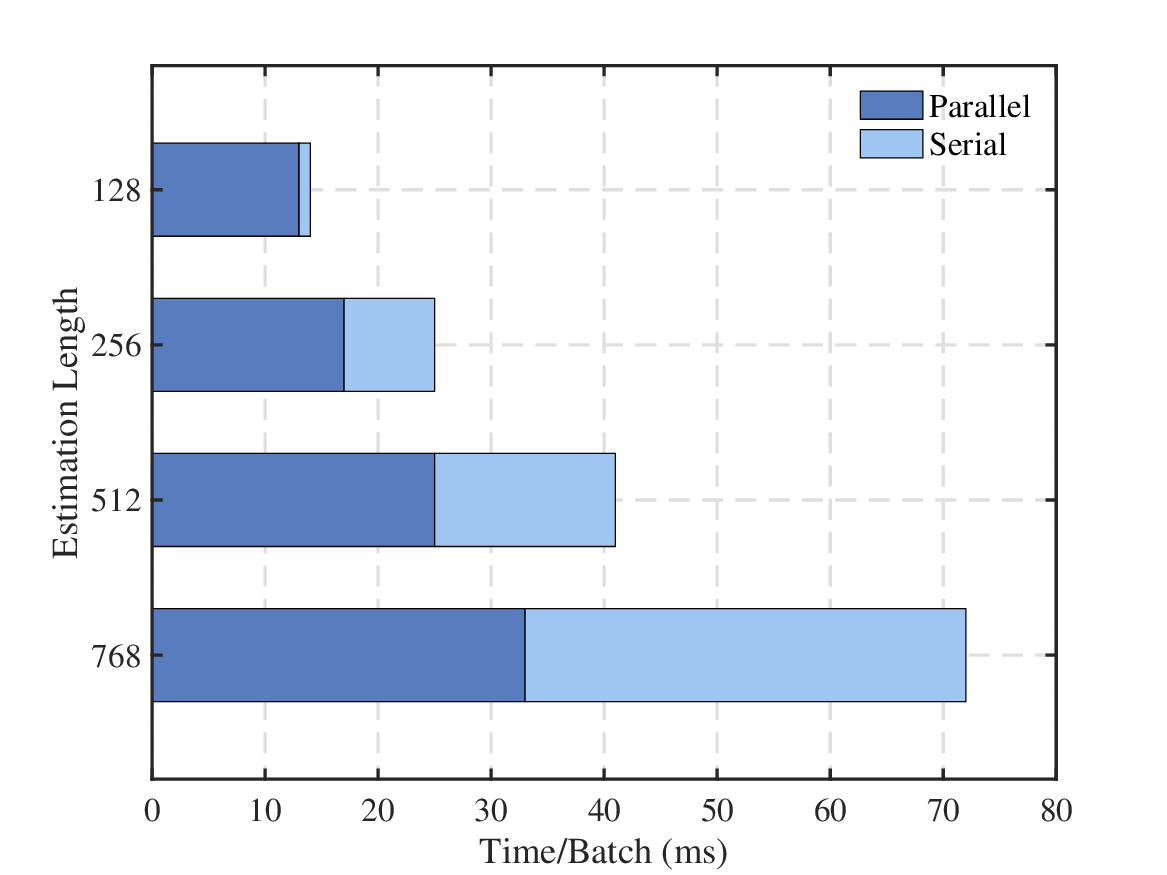}

  \caption{The training speed comparison between complexNDM parallel computing and serial computing on the experimental platform comprising an Intel(R) Xeon(R) Gold 6348 CPU @ 2.60GHz and an NVIDIA A800 GPU.}
 
  \label{fig: para}
\end{figure}

The parallelizable complexNDM can be accelerated using parallel processing units, such as GPUs and Tensor Processing Units (TPUs). Extensive research has been conducted on deploying parallel prefix sum algorithms on CUDA \cite{cuda}. However, custom CUDA operators may appear overly complex for industrial practitioners lacking familiarity with GPU principles and CUDA programming technology. To enhance method universality, we utilize JAX's \textit{jax.lax.associative\_scan} and custom binary operators to achieve parallel acceleration directly.

To validate the parallel acceleration effect, we adhered to the hyperparameters outlined in Table \ref{tab: hyper}. As shown in Fig. \ref{fig: para}, in benchmark experiments with an estimation length of 128, the training time for each batch computes in parallel (13ms) and serial (14ms) are tough equal. As the estimation length gradually increases, the acceleration effect becomes increasingly pronounced. When the estimation length reaches 768, the training speed is accelerated up to 2.2 times (72ms to 33ms). Therefore, the proposed method can be significantly accelerated with parallelization, which is favorable for large-scale and long-term prediction applications or online learning on the edge end (such as NVIDIA Jetson Orin Nano series) for real-time adaptive condition monitoring.

\section{Conclusion}\label{cons}
This paper proposed a parallelizable complex neural dynamics model for estimating PMSM temperature. This model is data-driven and based on SSMs. The parameterization of the state matrix in the complex domain enables a priori embedding of the stability and system oscillation frequency. Furthermore, we also bridge these physical priors above and the property of smooth evolution of the non-chaotic system. In addition, the diagonal decomposition of state matrices within the complex domain and the parallel prefix sum algorithm enable efficient hardware acceleration for complexNDM, achieving up to 2.2x acceleration in the training process. It validates the superiority of model performance using a real electric motor temperature experiment, which achieves an average RMSE of less than 1 $K$ with a compact model size.

\bibliographystyle{ieeetr}
\bibliography{ref}
\vspace{-1cm}
\begin{IEEEbiography}
[{\includegraphics[width=1in,height=1.25in,clip,keepaspectratio]{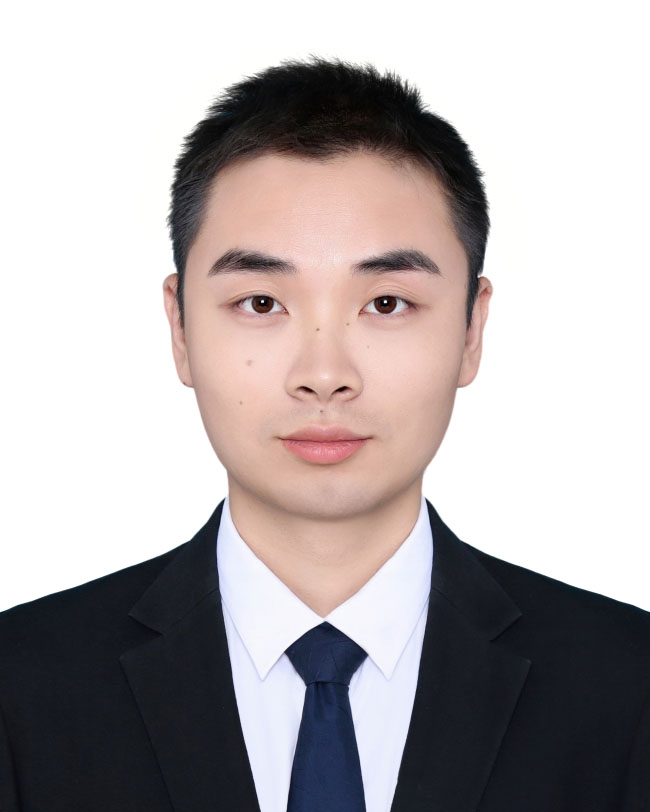}}]
{Xinyuan Liao} (S'24) received a B.Eng. degree in Computer Science and Technology from Ocean University of China in 2022 and a M.Eng. degree in Information and Communication Engineering from Northwestern Polytechnical University in 2025. He is currently working towards a PhD degree at the Department of Electrical and Electronic Engineering, The Hong Kong Polytechnic University, Hong Kong.

He was the recipient of the National Scholarship for Postgraduate at Northwestern Polytechnical University in 2024 and holds 5 Chinese invention patents. His current research interests include physics-informed machine learning and power electronics condition monitoring.
\end{IEEEbiography}

\vspace{-1cm}
\begin{IEEEbiography}[{\includegraphics[width=1in,height=1.25in,clip,keepaspectratio]{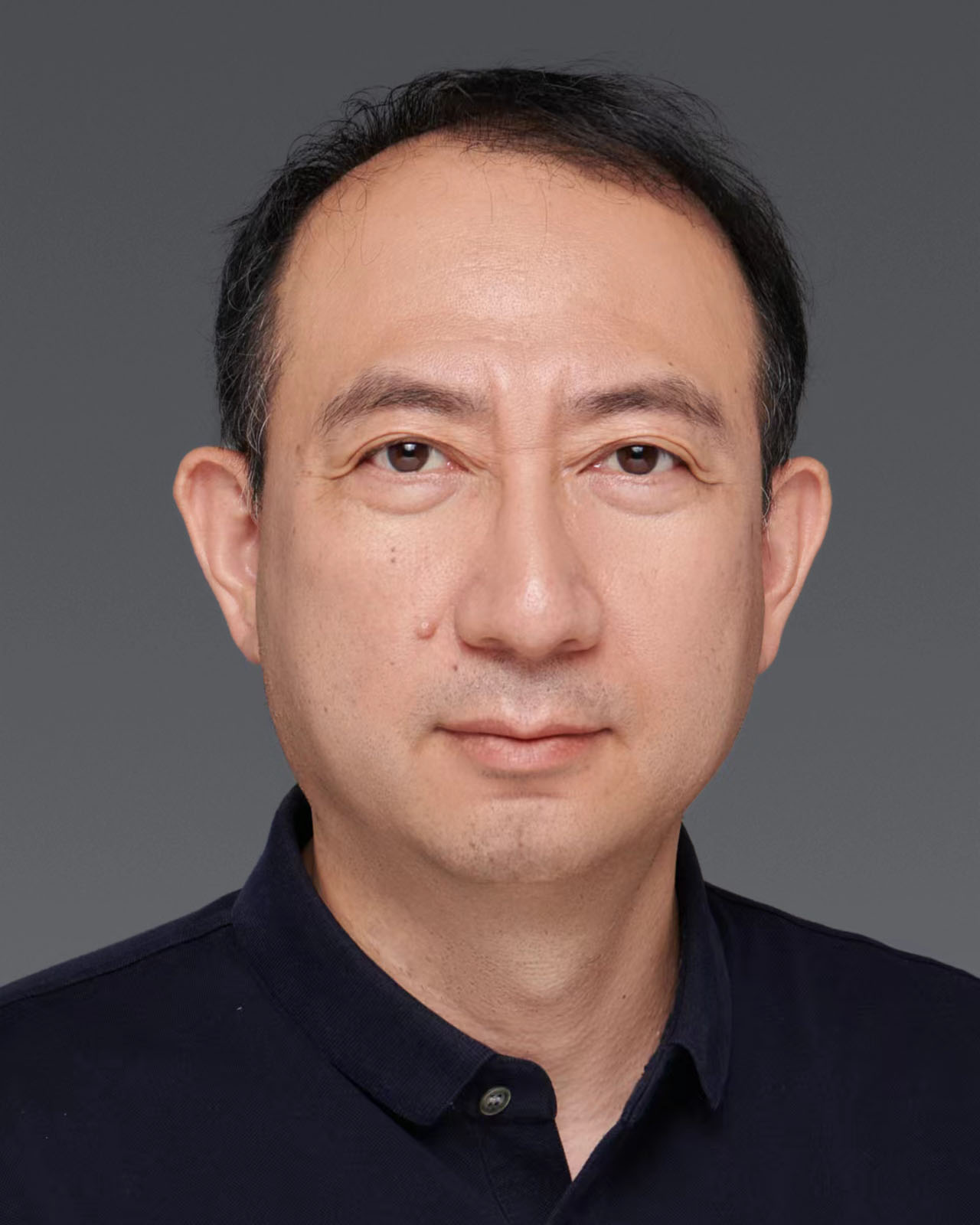}}]
{Shaowei Chen} (M'15) is currently an Associate Professor at the School of Electronics and Information, Northwestern Polytechnical University, Xi'an, China, where he is also the dean of the Department of Telecommunication Engineering and the Director of Perception and IoT Information Processing Laboratory. He is the principal investigator of several projects supported by the Aeronautical Science Foundation of China, Beijing, China. 

His research expertise is in the area of fault diagnosis, sensors, condition monitoring, and prognosis of electronic systems. Prof. Chen has been selected to receive several Provincial and Ministerial Science and Technology Awards. He is a senior member of the Chinese Institute of Electronics.
\end{IEEEbiography}

\vspace{-1cm}
\begin{IEEEbiography}[{\includegraphics[width=1in,height=1.25in,clip,keepaspectratio]{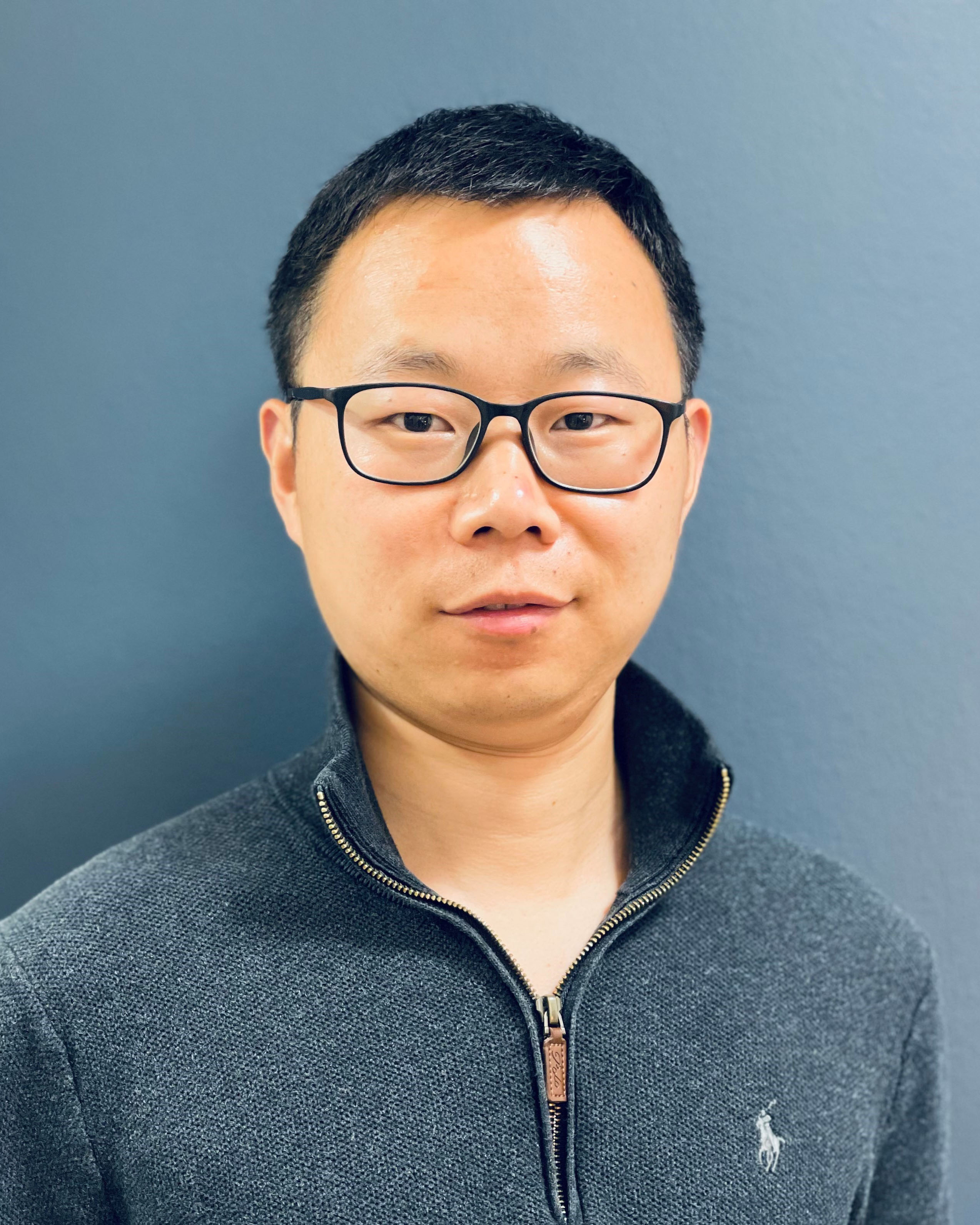}}]
{Shuai Zhao} (S'14-M'18-SM'24) received B.S. (Hons.), M.S., and Ph.D. degrees in Information and Telecommunication Engineering from Northwestern Polytechnical University, Xi'an, China, in 2011, 2014, and 2018, respectively. He is currently an Assistant Professor with AAU Energy, Aalborg University, Denmark.

From 2014 to 2016, he was a Visiting Ph.D. student at the University of Toronto, Canada. In August 2018, he was a Visiting Scholar with the University of Texas at Dallas, USA. From 2018 to 2022, he was a Postdoc researcher with AAU Energy, Aalborg University, Denmark. He is the Associate Editor of IEEE Transactions on Vehicle Technology and Guest Editor of IEEE Journal of Emerging and Selected Topics in Industrial Electronics and Elsevier e-Prime. His research interests include physics-informed machine learning, system informatics, condition monitoring, diagnostics \& prognostics, and tailored AI tools for power electronic systems.

\end{IEEEbiography}

\end{document}